\pdfoutput=1

\documentclass[english,letterpaper,hyperref,cleveref]{lipics-v2019}
\nolinenumbers
\usepackage{odonnell}
\usepackage{listings}% http://ctan.org/pkg/listings
\crefname{ineq}{inequality}{inequalities}
\creflabelformat{ineq}{#2{\upshape(#1)}#3}
\crefname{fact}{fact}{facts}
\crefname{definition}{definition}{definitions}

\newcommand\llfloor{\lfloor\!\!\lfloor}
\newcommand\rrfloor{\rfloor\!\!\rfloor}
\newcommand\apx[1]{\overset{#1}{\approx}}
\newcommand\notapx[1]{\overset{#1}{\not \approx}}

\title{Quantum Approximate Counting with Nonadaptive Grover Iterations}
\author{Ramgopal Venkateswaran\footnote{The ordering of the authors was randomized.  The authors contributed equally.}}{Computer Science Department, Carnegie Mellon University, USA}{ramgopav@andrew.cmu.edu}{}{}
\author{Ryan O'Donnell}{Computer Science Department, Carnegie Mellon University, USA \and \url{https://www.cs.cmu.edu/~odonnell} }{odonnell@cs.cmu.edu}{}{Supported by NSF grant CCF-1717606. This material is based upon work supported by the National Science Foundation under the grant number listed above. Any opinions, findings and conclusions or recommendations expressed in this material are those of the authors and do not necessarily reflect the views of the National Science Foundation (NSF).}%TODO mandatory, please use full name; only 1 author per \author macro; first two parameters are mandatory, other parameters can be empty. Please provide at least the name of the affiliation and the country. The full address is optional

\authorrunning{R. Venkateswaran and R. O'Donnell} %TODO mandatory. First: Use abbreviated first/middle names. Second (only in severe cases): Use first author plus 'et al.'

\Copyright{Ramgopal Venkateswaran and Ryan O'Donnell} %TODO mandatory, please use full first names. LIPIcs license is "CC-BY";  http://creativecommons.org/licenses/by/3.0/

 %TODO mandatory: Please choose ACM 2012 classifications from https://dl.acm.org/ccs/ccs_flat.cfm

 %TODO mandatory; please add comma-separated list of keywords

\category{} %optional, e.g. invited paper

\relatedversion{} %optional, e.g. full version hosted on arXiv, HAL, or other respository/website
\supplement{}%optional, e.g. related research data, source code, ... hosted on a repository like zenodo, figshare, GitHub, ...

\acknowledgements{The authors would like to thank Scott Aaronson and Patrick Rall for helpful communications.}%optional

\hideLIPIcs  %uncomment to remove references to LIPIcs series (logo, DOI, ...), e.g. when preparing a pre-final version to be uploaded to arXiv or another public repository

\EventEditors{John Q. Open and Joan R. Access}
\EventNoEds{2}
\EventLongTitle{42nd Conference on Very Important Topics (CVIT 2016)}
\EventShortTitle{CVIT 2016}
\EventAcronym{CVIT}
\EventYear{2016}
\EventDate{December 24--27, 2016}
\EventLocation{Little Whinging, United Kingdom}
\EventLogo{}
\SeriesVolume{42}
\ArticleNo{23}
\begin{document}

\maketitle

\begin{abstract}
    Approximate Counting refers to the problem where we are given query access to a function $f : [N] \to \{0,1\}$, and we wish to estimate $K = \#\{x : f(x) = 1\}$ to within a factor of $1+\epsilon$ (with high probability), while minimizing the number of queries.
    In the quantum setting, Approximate Counting can be done with $O\bparens{\min\bparens{\sqrt{N/\epsilon}, \sqrt{N/K}\!\bigm/\!\eps}}$ queries.
    It has recently been shown that this can be achieved by a simple algorithm that only uses ``Grover iterations''; however the algorithm performs these iterations adaptively.
    Motivated by concerns of computational simplicity, we consider algorithms that use Grover iterations with limited adaptivity.
    We show that algorithms using only nonadaptive Grover iterations can achieve $O\bparens{\sqrt{N/\eps}}$ query complexity, which is tight.
\end{abstract}

\newpage

\section{Introduction}
\subsection{Grover Search recap}
A famous, textbook algorithm in quantum computing is \emph{Grover Search}~\cite{Gro96}, which solves the following task:
Given is a quantum oracle for a function $f : [N] \to \{0,1\}$, where queries for $f(x)$ may be made in quantum superposition.
It is promised that $K = \#\{x : f(x) = 1\}$ is exactly~$1$.
The task is to find $x^*$ such that $f(x^*) = 1$.
Grover Search solves this problem (with high probability) using $O(\sqrt{N})$ queries.

The algorithm is particularly simple:
First, a state $\ket{s}$ equal to the uniform superposition over all $\ket{x}$ is prepared; this state makes an angle of $\arccos\sqrt{1/N}$ with $\ket{x^*}$.
We write $\ket{x^*}^\bot$ for the state perpendicular to $\ket{x^*}$ making an angle of $\theta^* = \arcsin \sqrt{1/N}$ with~$\ket{s}$.
Then the algorithm repeatedly performs \emph{Grover iterations}, each of which consists of one query followed by the simple ``Grover diffusion'' operation.
The effect of a Grover iteration is to rotate $\ket{s}$ by an angle of~$2 \theta^*$; thus after $r$~rotations the angle of $\ket{s}$ from $\ket{x^*}^\bot$ is $(2r+1)\theta^*$.
Setting $r \approx \tfrac{\pi}{4} \sqrt{N}$, the algorithm makes $O(\sqrt{N})$ queries and ends up with a state at angle approximately $\tfrac{\pi}{2}$ from $\ket{x^*}^\bot$; measuring then results in~$\ket{x^*}$ with probability $\sin^2((2r+1)\theta^*) \approx \sin^2 \tfrac{\pi}{2} = 1$.

In this form the algorithm relies on the assumption $K = 1$.
For $K \geq 1$, the only change is that the parameter $\theta^*$ becomes $\arcsin \sqrt{K/N}$.
Thus if the correct value of $K$ is known to the algorithm (and we assume for simplicity that $K \leq N/2$), it can choose $r \approx \tfrac{\pi}{4} \sqrt{N/K}$ and solve the search problem using $O(\sqrt{N/K})$ iterations/queries.
This algorithm also works if the algorithm knows an estimate $K'$ of $K$ that is correct up to small multiplicative error; say, $K' \apx{1.1} K$.
Here we are using the following notation:
\begin{notation}
    For $a, b, \eta > 0$ we write $a \apx{1+\eta} b$ if $\frac{1}{1+\eta} \leq a/b \leq 1+\eta$.
\end{notation}

If $K$ is \emph{unknown} to the algorithm, one possibility is try all estimates $K' = 1$, $1.1$, $1.1^2$, $1.1^3,\  \dots,\ N/2$.
The total number of Grover iterations (hence queries) will be $O(\sum_{i} \sqrt{N/1.1^i}) = O(\sqrt{N})$.\footnote{One must take a small amount of care to bound the overall failure probability without incurring a log factor.}
A strategy with improved query complexity for large~$K$ was given by Boyer, Brassard, H\o{}yer, and Tapp~\cite{BBHT98}; in brief, for $i = 0, 1, 2, \dots$ it tries a random number of rotations in the range $[1,1.1^i]$, stopping if ever an $x \in f^{-1}(1)$ is found.
This algorithm solves the search problem using $O(\sqrt{N/K})$ iterations, despite not knowing~$K$ in advance.

\subsection{Approximate Counting}   \label{sec:ac}
A natural related problem, called \emph{Approximate Counting} and introduced in~\cite{BBHT98}, is to \emph{estimate}~$K$.
More precisely, given as input a parameter $\eps \geq 0$, the task is to output a number $\wh{K}$ such that (with high probability) $\wh{K} \apx{1+\eps} K$.
To keep the exposition simple, in this paper we will make the following standard assumptions:
\begin{itemize}
    \item $1/N \leq \eps \leq 1$ (setting $\eps = 1/N$ just yields the problem of exact counting --- any smaller value of~$\eps$ does not change the problem);
    \item $K \leq N/2$ (otherwise there is generally a dependence on $N-K$, since one can switch the roles of $0$ and $1$ in $f$'s output);
    \item $K \neq 0$ (generally, all algorithms can easily be extended to work in the case of~$K = 0$, with query complexity being the worst-case query complexity over all $K > 0$).
\end{itemize}
The quantum Approximate Counting problem was solved with optimal query complexity by Brassard, H\o{}yer, Mosca, and Tapp~\cite{BHMT02}.
Combining quantum Fourier transform ideas from Shor's Algorithm with the ideas behind Grover Search, their algorithm solves Approximate Counting using $O\parens{\sqrt{N/K}\!\bigm/\!\eps}$ queries.

Let us make some remarks about this query complexity.
First, note that the bound takes $K$ into account, even though~$K$ is (initially) unknown to the algorithm.
Second, although $K$ could be as small as~$1$, the worst-case query complexity over all~$K$ need \emph{not} be $\Omega\parens{\sqrt{N}\!\bigm/\!\eps}$.
(Indeed, this would lead to an illogical query complexity of $\Omega(N^{3/2})$ if one set $\eps = 1/N$ to do exact counting.)
Instead, note that an algorithm can first run the \cite{BHMT02} algorithm with $\eps = 1$, expending $O\parens{\sqrt{N/K}} \leq O(\sqrt{N})$ queries and learning a preliminary estimate $K' \apx{2} K$.
Now since~$K$ is an integer, there is no point in trying to approximate it to a factor better than~$1+1/K$, hence better than~$1 + 1/(2K')$.
Thus the algorithm can now raise the initial input~$\eps$ to $1/(2K')$ if necessary, and \emph{then} run \cite{BHMT02} to obtain its final estimate.
This yields a final query complexity of
\[
    O\parens{\sqrt{N/K}\!\bigm/\!\max\{\eps, 1/K\}} = \min\{O(\sqrt{NK}), O\parens{\sqrt{N/K}\!\bigm/\!\eps}\} \leq O(\sqrt{N/\eps})
\]
(where in the last inequality we took the geometric mean).
This $K$-independent bound of $O(\sqrt{N/\eps})$ is logical: the smallest $\eps$ one should ever take is $\eps = 1/N$, and this leads to a query complexity of $O(N)$ for the general case of exact counting.
By similar reasoning one can obtain the more precise fact that exact counting can done with $O(\sqrt{NK})$ queries.
Finally, we remark that \cite{BHMT02}'s query complexity was shown to be optimal by Nayak and Wu~\cite{NW99}.

Let us also briefly mention the quantum \emph{Amplitude Estimation} problem, which is essentially the same as the Approximate Counting problem except that the ``initial angle'' $\theta^*$ need not be of the form $\arcsin\sqrt{K/N}$ for some integer~$K$, but can be any value.
The solution to the Amplitude Estimation problem in~\cite{BHMT02} is a widely used tool in quantum algorithm design, and leads to  quadratic speedups over classical algorithms for a variety of statistical problems.

\subsection{Simpler and nonadaptive?}
Although the Approximate Counting algorithm from~\cite{BHMT02} has optimal query complexity, there has recently been a lot of interest in simplifying it~\cite{Wie19,SURTOY20,AR20}.
In particular the latter two of these just-cited works strove to replace it with an algorithm that \emph{only} uses Grover iterations, both for analytic simplicity and practical simplicity (the controlled amplifications of~\cite{BHMT02} being particularly problematic for NISQ devices).
The work of Aaronson and Rall~\cite{AR20} provably succeeds at this challenge, providing an algorithm that solves the Approximate Counting problem using $O\parens{\sqrt{N/K}\!\bigm/\!\eps}$ Grover iterations (hence queries).
Briefly, the Aaronson--Rall algorithm has a first phase (somewhat similar to the \cite{BBHT98} algorithm) that performs a geometrically increasing sequence of Grover iterations until $K$ can be estimated up to a constant factor of~$1.1$ (see \Cref{thm:AR} herein).
This requires $O(\sqrt{N/K})$ iterations.
In the second phase, their algorithm performs a kind of binary search to improve the approximation factor to~$1+\eps$; each step of the binary search requires additional Grover iterations, totalling $O\parens{\sqrt{N/K}\!\bigm/\!\eps}$ in the end.

From the point of view of practicality and simplicity, there is a downside to the Aaronson--Rall algorithm, which is that its Grover iterations are \emph{adaptive} (especially in the second phase of the algorithm).
In other words, the steps of the algorithm involve many repetitions of the following: performing some Grover iterations, measuring, and doing some classical computation to decide how many Grover iterations to do in the next step.
It has been argued that this repeated switching between quantum and classical computation could be undesirable in practice.
Indeed, the final open question in~\cite{AR20} concerned the optimal query complexity of Approximate Counting using \emph{nonadaptive} Grover iterations.
This version of the problem was also stressed and studied in~\cite{SURTOY20}, but without any provable guarantees being provided. 

Other recent developments in the area of approximate counting include \cite{grinko2020iterative,nakaji2020faster}, which propose variants of the algorithm from~\cite{AR20} but with improved constant factors. As well, the work~\cite{burchard2019lower} proves a lower bound for the query complexity of approximate counting in the parallel case.

\subsection{Our results}    \label{sec:our}
We investigate the problem of Approximate Counting using only nonadaptive Grover iterations.
Note that for this version of the problem, there is no hope of obtaining the query complexity  $O\parens{\sqrt{N/K}\!\bigm/\!\eps}$ that improves as a function of~$K$.
To see this, suppose even that $\eps$ is fixed to~$1$.
If the algorithm is to achieve query complexity $O(\sqrt{N/K})$, then it must be able to achieve $O(1)$ query complexity when $K = \Theta(N)$.
Since it is nonadaptive, this means it must \emph{always} make only $O(1)$ queries.
But this is impossible, as even for \emph{adaptive} algorithms it is known that $\Omega(\sqrt{N})$ queries are required in the case of $K = O(1)$, $\eps = 1$.

In other words, with nonadaptive algorithms we can only hope to achieve the optimal query complexity that is independent of~$K$, namely $O(\sqrt{N/\eps})$.
In this work we indeed show this is achievable.
Our main theorem is:
\begin{theorem}                                     \label{thm:main}
    There is an algorithm for  quantum Approximate Counting  that uses only \emph{nonadaptive} Grover iterations, and that has a query complexity of \mbox{$O(\sqrt{N/\eps})$} (and minimal additional computational overhead).
\end{theorem}

We also briefly sketch an extension of our algorithm achieving improved query complexity in the setting where we are allowed multiple rounds of nonadaptive Grover iterations (as opposed to just one).

\section{Preliminaries}
We will assume throughout that $K \leq 2^{-20}N$.\footnote{Our work would be fine with, say, $K \leq N/8$, but we put $2^{-20}$ so as to able to cite~\cite{AR20} as a black box.}
This without loss of generality since we may artificially replace $N$ by $2^{20}N$ and extend $f$ to $f : [2^{20} N] \to \{0,1\}$, with $f(x) = 0$ for $x > N$.
We fix
\[
    \theta^* = \arcsin\sqrt{K/N},
\]
sometimes called the ``Grover angle''.
Recall that a query algorithm based on Grover iterations has the following property:
At the cost of $q \in \N$ ``queries'', it can ``flip a coin with bias $\sin^2((2q+1) \theta^*)$''.
By repeating this $t$ times, it can obtain~$t$ independent flips of this coin.
It is statistically sufficient to retain only the average of the coin flip outcomes, which is a random variable distributed as $\frac{1}{t}\text{Bin}(t, \sin^2((2q+1) \theta^*))$.
These observations lead to the following:
\begin{notation}
    For real $r \geq 1$ we write $\llfloor r \rrfloor$ for the largest odd integer not exceeding~$r$, and we write $p(r) = \sin^2( \llfloor r \rrfloor \theta^*)$.
\end{notation}
\begin{definition}
\label{definition: grover-schedule}
    A \emph{Grover schedule} consists of two sequences: $R = (r_1, \dots, r_m)$ (each real $r_i \geq 1$) and $T = (t_1, \dots, t_m)$ (each $t_i \in \N^+$).
    \emph{Performing} this Grover schedule refers to obtaining independent random variables $\wh{\bp}_1, \dots, \wh{\bp}_m$, where $\wh{\bp}_i$ is distributed as $\frac{1}{t_i}\text{Bin}(t_i, p(r_i))$.
\end{definition}

A nonadaptive Grover iteration algorithm for Approximate Counting is simply an algorithm that performs one fixed Grover schedule, and produces its final estimate~$\wh{K}$ by classically post-processing the results.
One can more generally study algorithms with ``$s$~rounds of nonadaptivity''; this simply means that $s$~Grover schedules are used, but they may be chosen adaptively.
\begin{fact}    \label{fact:schedule-calcs}
    Performing the Grover schedule $R, T$ uses at most $\tfrac12\sum_i r_i t_i$ queries.
\end{fact}

\subsection{On how well we need to approximate $\theta^*$}

We will use the following elementary numerical fact:
\begin{lemma}                                       \label{lem:close}
    Suppose for real $0 \leq k,k' \leq N$ and $\eta \leq 1$ that $\arcsin\sqrt{k'/N} \apx{1+\eta} \arcsin\sqrt{k/N}$.
    Then $k' \apx{1+3\eta} k$.
\end{lemma}
This lemma helps us show that approximating $\theta^*$ well is equivalent to approximating~$K$ well:
\begin{proposition}                                     \label{prop:closeK}
    Suppose $\theta \apx{1+\eps/6} \theta^*$.
    If $\kappa' \in \R$ satisfies $\theta = \arcsin\sqrt{\kappa'/N}$, and $K'$ is the nearest integer to~$\kappa'$, then $K' \apx{1+\eps} K$.
\end{proposition}
\begin{proof}
    Since  $\theta \apx{1+\eps/6} \theta^*$, \Cref{lem:close} tells us that $\kappa' \apx{1+\eps/2} K$, and hence
    \begin{equation}    \label[ineq]{ineq:whoa}
        |\kappa' - K| \leq (\eps/2) K.
    \end{equation}
    We also have $|\kappa' - K'| \leq 1/2$, and hence $|K' - K| \leq (\eps/2) K + 1/2$.
    But we can assume $(\eps/2) K \geq 1/2$, as otherwise \Cref{ineq:whoa} implies $|\kappa' - K| < 1/2$ and hence $K' = K$.
    Thus $|K' - K| \leq (\eps/2)K + (\eps/2) K = \eps K$; i.e., $K' \apx{1+\eps} K$.
\end{proof}

\begin{lemma}                                     \label{lem:reducedEps}
    Given some $\theta' \apx{1.11} \theta^*$, estimating $\theta^*$ to a factor of $1 + 1/(2N\sin^2{\theta'})$ is at least as good as estimating $\theta^*$ to a factor of $1 + \eps / 6$.
\end{lemma}

\begin{proof}
    From \Cref{lem:close}, $1 + 1/(2N\sin^2{\theta'}) \leq 1 + 1.33/(2K) < 1 + 1/K$. The closest possible values to $K$ are $K - 1$ and $K + 1$; therefore, estimating $K$ within a factor of $1 + 1/K$ is the same as estimating $K$ exactly. This is at least as good as estimating $K$ to within a factor of $1 + \eps / 6$.
\end{proof}

\section{The nonadaptive algorithm}
Our algorithm can conceptually be thought of as having two stages: the first stage estimates~$\theta^*$ to a constant factor, and the second stage improves this estimate to the desired factor of $1 + \eps$. This two-stage approach is similar in flavor to the algorithms in \cite{AR20,BBHT98}. However we note that, consistent with our  nonadaptivity condition, the two stages in our algorithm can be run in parallel.

For the first stage of our algorithm, we require the following result of Aaronson and Rall~\cite{AR20}, which estimates $\theta^*$ up to a factor of~$1.1$, using $O(\sqrt{N})$ nonadaptive queries.
(In fact, as Aaronson and Rall show, the obvious adaptive version of the algorithm incurs only $O(\sqrt{N/K})$ queries.)
\begin{theorem}                                     \label{thm:AR}
    Let $R = (1, (12/11), (12/11)^2, \dots, (12/11)^m)$, where $m = \Theta(\log N)$ is minimal with $(12/11)^m \geq \sqrt{N}$.
    Let $T$ consist of $m$ copies of $10^5 \ln(120/\delta)$.
    Perform the Grover schedule $R, T$.
    (By \Cref{fact:schedule-calcs} this incurs $O(\sqrt{N})$ queries.)
    Then except with probability at most~$\delta/2$, there is a minimal $t$ such that $\wh{\bp}_t \geq 1/3$, and setting $\wt{\theta} = (5/8)(11/12)^{t}$ results in $\wt{\theta} \apx{1.1} \theta^*$.
\end{theorem}

The second stage of our algorithm uses the following critical lemma, which we will prove in \Cref{section:provingMainLemma}.

\newtheorem*{L9}{Lemma~\ref{lem:estimateToEps}}

\begin{lemma}                                     \label{lem:estimateToEps}
    Given the parameters $\theta’$, $\eps’$, and $\delta’$, there is an algorithm using only nonadaptive Grover iterations that performs $O(\log(1/\delta')/(\theta'\eps'))$ queries, and outputs a result $\theta_{\mathrm{est}}$ with the following guarantee: if $\theta’ \apx{1.11} \theta^*$, then $\theta_{\mathrm{est}} \apx{1 + \eps'/6} \theta^*$ except with probability at most $\delta/2$.
\end{lemma}

In this section, we will show how to use \Cref{thm:AR} and \Cref{lem:estimateToEps} to prove \Cref{thm:main}.

We will now state our algorithm:
\begin{algorithm}   \label{alg:1}\hphantom{placeholderToEndLine}
\begin{enumerate}
    \item Run the Aaronson--Rall algorithm from \Cref{thm:AR}, allowing us to later compute $\wt{\theta}$.
    \item For $\theta = \arcsin(\sqrt{1/N}), 1.001\arcsin{\sqrt{1/ N}}, (1.001)^2\arcsin{\sqrt{1/ N}}, \dots, 1.1\arcsin(2^{-20})$:
    
    \quad Perform the algorithm in \Cref{lem:estimateToEps} with the parameters $\theta' = \theta$, $\eps' = \max(\eps, \frac{1}{2N\sin^2\theta'})$, and $\delta' = \delta/2$.
    \item Classical Post-processing: Among all iterations in the for-loop, take the iteration with the value of $\theta$ that was closest to~$\wt{\theta}$, and output the result of that iteration.
\end{enumerate}
\end{algorithm}

Each iteration of the for loop in Step~$2$ can be done in parallel (there are no computational dependencies between the iterations), and Step~$1$ can also be done in parallel with Step~$2$. Therefore, the algorithm uses only nonadaptive Grover iterations. Note also that we can write the quantum parts of steps $1$ and $2$ as one fixed Grover schedule, with the classical parts and step $3$ forming the post-processing step; however, it will be more convenient in this section to think about these as individual steps in a logical sequence.

\begin{proposition}
\label{prop:overallCorrectness}
    \Cref{alg:1} returns a value $\theta_{\mathrm{est}}$ such that $\theta_{\mathrm{est}} \apx{1 + \eps} \theta^*$, except with probability at most $\delta$.
\end{proposition}
\begin{proof}
    By \Cref{thm:AR}, Step~$1$ returns an estimate $\wt{\theta}$ such that $\wt{\theta} \apx{1.1} \theta^*$, with a failure probability of at most $\delta/2$ . The value of $\theta$ (in Step~$2$) that is closest to $\wt{\theta}$ is at most a factor of $1.001$ away from $\wt{\theta}$. If Step~$1$ succeeded, this is at most a factor of $1.1 \times 1.001 < 1.11$ away from $\theta^*$. By \Cref{lem:estimateToEps}, the algorithm outputs an estimate $\theta_{\mathrm{est}}$ such that $\theta_{\mathrm{est}} \apx{1 + \eps'/6} \theta^*$, with a failure probability of at most $\delta/2$. \Cref{lem:reducedEps} then implies that $\theta_{\mathrm{est}} \apx{1 + \eps/6} \theta^*$. Using \Cref{prop:closeK} (setting the parameter $\theta = \theta_{\mathrm{est}}$), we get an estimate $K_{\mathrm{est}}$ such that $K_{\mathrm{est}} \apx{1 + \eps} K$. By the union bound, the overall failure probability of the algorithm is at most $\delta$.
\end{proof}

\begin{proposition}                                        
\label{prop:overallQueryComplexity}
    \Cref{alg:1} makes $O(\sqrt{N/\eps}\log({1}/{\delta}))$ queries.
\end{proposition}
\begin{proof}
    First, consider all iterations where $2N\sin^2(\theta) \leq 1/\eps$. In these cases, the query complexity given by \Cref{lem:estimateToEps} would be
    $O(\log(1/\delta)(N\sin^2(\theta))/\theta) = O(N\theta\log(1/\delta))$.
    
    The query complexity associated with these iterations is a geometric series with a constant common ratio of $1.01$ where the largest term is $O(\sqrt{N/\eps}\log(1/\delta))$. Therefore the overall query complexity due to these iterations is $O(\sqrt{N/\eps}\log(1/\delta))$.
    
    Now consider all iterations where $2N\sin^2(\theta) > 1/\eps$. In these cases, the query complexity is $O(\log(1/\delta)/(\theta\eps))$. This forms a geometrically decreasing series (with a constant common ratio), where the first term is again $O(\sqrt{N/\eps}\log(1/\delta))$. The overall query complexity contributed by these schedules is thus also $O(\sqrt{N/\eps}\log(1/\delta))$.
    
    Therefore, the query complexity of \Cref{alg:1} is $O(\sqrt{N/\eps}\log(1/\delta))$, as claimed.
\end{proof}
Having proven \Cref{prop:overallCorrectness} and \Cref{prop:overallQueryComplexity}, we have established our main result \Cref{thm:main} modulo the proof of \Cref{lem:estimateToEps}, which will appear in the next section.  Before giving this, we briefly sketch how our algorithm can also be extended to the setting of being allowed multiple rounds of nonadaptive Grover iterations. If we have two such rounds of nonadaptivity, we can first run step $1$ of our algorithm to get a constant-factor approximation, and then based on its result run the algorithm in \Cref{lem:estimateToEps}; this achieves a query complexity of $O(\sqrt{N} + \min(\sqrt{N/\eps}, \sqrt{N/K}/\eps))$. This nearly matches the query complexity of the fully adaptive case, but for the $\sqrt{N}$ term due to the first step. Given more rounds of nonadaptivity, we can reduce the cost of this first step by staging it over multiple initial rounds. One can show that with $O(\log N)$ rounds of nonadaptivity, this will yield the optimal query complexity corresponding to the fully adaptive case.

\section{Proving \Cref{lem:estimateToEps}}
\label{section:provingMainLemma}

Our goal in this section will be to prove \Cref{lem:estimateToEps}. Assume that we are given some $\theta'$, $\eps'$, and $\delta'$. We will show a nonadaptive Grover iteration algorithm making $O(\log(1/\delta')/(\theta'\eps'))$ queries with the property that if $\theta' \apx{1.11} \theta^*$, then its output will be a factor-$(1 + \eps')$ approximation of $\theta^*$ (except with failure probability at most $\delta'$) For the remainder of the section, we will assume that we are in the interesting case where $\theta' \apx{1.11} \theta^*$ (in the other cases, the algorithm does not need to output a correct answer).

\subsection{Proof idea}
\label{subsection:proofIdea}
The algorithm for \Cref{lem:estimateToEps} is structured exactly as described in \Cref{definition: grover-schedule} and \Cref{fact:schedule-calcs}; there is an initial nonadaptive quantum part with a fixed Grover schedule (that we will later define), and a classical post-processing step at the end that uses the results of the quantum part to estimate $\theta^*$.

Before stating the key ideas in the quantum part of our algorithm, we mention the ``Rotation Lemma'' of Aaronson and Rall~\cite[Lem.~2]{AR20}. The main idea in that lemma can be roughly stated as follows: given that $\theta^*$ lies in some range $[\theta_{\mathrm{min}}, \theta_{\mathrm{min}} + \Delta\theta]$, we can pick an odd integer value of $r$ (where  $r = O(1/(\theta \cdot \Delta\theta))$), such that $r\theta_{\mathrm{min}}$ is close to $2\pi k$ and $r(\theta_{\mathrm{min}} + \Delta\theta)$ is close to $2\pi k + \pi/2$. If $\theta$ is close to $\theta_{\mathrm{min}}$, $p(r)$ will be nearly $0$ (and if it is close to $\theta_{\mathrm{min}} + \Delta\theta$, it will be nearly $1$). Aaronson and Rall use this lemma to continually shrink the possible range that $\theta^*$ could lie in by a geometric factor at each iteration, until the range is $1 \pm \epsilon$.

We will adopt a similar idea to find an efficient Grover schedule that can distinguish any two candidate angles with high probability; we do this by relaxing the condition of one angle being close to $2\pi k$ and the other being at distance $\pi/2$ from it. Instead, we choose the sequence $R$ in our Grover schedule such that for any pair of values $\theta_1$, and $\theta_2$, there is some $r \in R$ such that $r\theta_1$ and $r\theta_2$ differ by approximately $\pi/8$, and are also ``in the same quadrant'' (meaning the same interval $[0,\tfrac{\pi}{2})$, $[\tfrac{\pi}{2}, \pi)$, $[\pi, \tfrac{3\pi}{2})$, $[\tfrac{3\pi}{2}, 2\pi)$ modulo~$2\pi$). This relaxation allows us to save on the total number of queries made by reusing the same value of $r$ to distinguish many pairs of candidate angles. Due to this, the nonadaptivity requirement does not make the query complexity grow polynomially larger (whereas, for example, naively simulating the search tree from \cite{AR20} in a nonadaptive fashion would incur an extra $1/\eps'$ factor).

The classical post-processing involves running a ``tournament'' between all candidate estimates of $\theta^*$, which outputs the winning value as the estimate. This post-processing step can be implemented efficiently in $O(\log(1/\eps')/\eps')$ classical time.

\subsection{Some arithmetic lemmas}

We now define some useful sequences and prove a couple of arithmetic lemmas about them.

\begin{definition}
    Define  the sequence $u$ by $u_0 = 1$, $u_1 = 1.01$, \dots, $u_L = 1.01^L$ where $L = O(\log(1/(\theta'\eps')))$ is minimal with $u_L \geq 1.2\pi/(\theta'\eps')$.
\end{definition}

\begin{lemma}                                       \label{lem:split}
    Suppose we are given $\theta_0$ and $\theta_1$ such that $\theta_0 < \theta_1$, $\theta_0 \apx{1.11} \theta'$, $\theta_1 \apx{1.11} \theta'$, and $\theta_0 \notapx{1+\eps/6.1} \theta_1$. Write $\eta = \theta_1 - \theta_0$. Then there exists some $0 \leq i \leq L$ such that $u_i\eta \apx{1.01} \tfrac{\pi}{8}$.
\end{lemma}
\begin{proof}
    We know that $u_0\eta = \eta \leq \theta_1 \leq 1.1 \times 0.0001 < \tfrac{\pi}{8}$. We also have $u_L\eta \geq 1.2\pi\eta/(\theta'\eps') \geq 1.2\pi\eta/(1.1\theta_0\eps') \geq 1.2\pi/(1.1 \cdot 6.1) > \tfrac{\pi}{8}$ where we used $\theta_0 \notapx{1+\eps/6.1} \theta_1$ in the second-to-last inequality. The lemma now follows from the geometric growth of the $u_i$'s with ratio $1.01$. In particular, $i = \floor*{\log_{1.01}(\tfrac{\pi}{8\eta})}$ works.
\end{proof}

For the $u_i$ given by \Cref{lem:split}, we have $u_i \theta_1 - u_i \theta_0 \approx \frac{\pi}{8}$.
This seems promising, in that the  ``coin probabilities'' associated to these angles, namely $\sin^2(u_i \theta_1)$ and $\sin^2(u_i \theta_0)$, seem as though they should be far apart.
Unfortunately, something annoying could occur; it could be that these angles are, say, $100\pi \pm \frac{\pi}{16}$, in which case the coin probabilities would be identical.
As mentioned in \Cref{subsection:proofIdea}, what we would \emph{really} like is to have the two angles be far apart but also in the same quadrant.
To achieve this, we will define a new sequence.

\begin{definition}
    For each $0 \leq i \leq L$, define $a_{i, 0} = 0$, $a_{i, 1} = \frac{\pi}{4.8\theta'}$, $a_{i, 2} = 1.01 \cdot \frac{\pi}{4.8\theta'}$, $a_{i, 3} = 1.01^2 \cdot \frac{\pi}{4.8\theta'}$, \dots, $a_{i, C + 1} = 1.01^{C} \cdot \frac{\pi}{4.8\theta'}$, where $C = \ceil*{2\log_{1.01}(1.2)}$ is a constant. Also define $s_{i, j} = u_i + a_{i, j}$.
\end{definition}

\begin{lemma}                                       \label{lem:refine}
    In the setting of \Cref{lem:split}, there exists some $0 \leq j \leq C + 1$ such that
    \[
        s_{i,j} \eta  \apx{1.5} \tfrac{\pi}{8}
    \]
    and such that $s_{i,j} \theta_0$ and $s_{i,j} \theta_1$ are in the same quadrant.
\end{lemma}
\begin{proof}
    We first apply \Cref{lem:split} and obtain
    \begin{equation}    \label[ineq]{ineq:clo}
        u_i \eta \apx{1.01} \tfrac{\pi}{8}.
    \end{equation}
    Now if $u_i \theta_0$ and $u_i \theta_1$ are already in the same quadrant then we can take $j = 0$ (implying $s_{i,j} = u_i$) and we are done.
    Otherwise,  the plan will be to find $j > 0$ with $a_{j} \theta_0 \approx \tfrac{\pi}{4}$, thus shifting them to $s_{i, j} \theta_0$ and $s_{i,j} \theta_1$ that still differ by roughly $\tfrac{\pi}{8}$ but which now must be in the same quadrant.

    To find the required~$j$, observe that on one hand, $a_{1} \theta_0 = (\pi/(4.8\theta')) \cdot \theta_0 \leq 1.11 \pi / 4.8 \leq \frac{\pi}{4}$.
    On the other hand, $a_{i,C + 1} \theta_0 \geq 1.2^2 \cdot (\pi/(4.8\theta')) \cdot \theta_0 \geq (1.2/1.1) \cdot \pi/4 \geq \frac \pi 4$.
    By the geometric growth of the $a_{i,j}$'s with ratio~$1.01$, we conclude that there exists some $1 \leq j \leq \ell_i$ achieving
    \begin{equation}    \label[ineq]{ineq:a}
        a_{i,j} \theta_0 \apx{1.05} \tfrac{\pi}{4}.
    \end{equation}
    We may now make several deductions.
    First,
    \[
        \theta_0 \apx{1.1} \theta_t, \theta_1 \apx{1.1} \theta_t \medspace \implies \medspace
        \theta_0 \apx{1.22} \theta_1 \medspace \implies \medspace 
        a_{i,j} \theta_0 \apx{1.22} a_{i,j} \theta_1 \medspace \implies \medspace
        a_{i,j} \eta \leq .22 \cdot a_{i,j} \theta_0 \leq .22 \cdot 1.05 \tfrac{\pi}{4} \leq .24 \tfrac{\pi}{4}.
    \]
    Combining this with \Cref{ineq:clo} we conclude
    \begin{equation}    \label[ineq]{ineq:b}
        s_{i,j} \eta  = u_{i}\eta  + a_{i,j}\eta  \in [\tfrac{1}{1.01} \tfrac{\pi}{8}, \tfrac{\pi}{8}(1.01) + .24\tfrac{\pi}{4}] \apx{1.5} \tfrac{\pi}{8}.
    \end{equation}
    Thus we started with $u_i \theta_0$ and $u_i \theta_1$ differing by $\tfrac{\pi}{8}$ (up to factor~$1.01$) but in different quadrants; by passing to $s_{i,j} \theta_0$ and $s_{i,j} \theta_1$, we have offset $u_i \theta_0$ by $\tfrac{\pi}{4}$ (up to factor~$1.05$, \Cref{ineq:a}) and the two angles still differ by around $\tfrac{\pi}{8}$ (up to factor~$1.5$, \Cref{ineq:b}).
    Thus $s_{i,j} \theta_0$ and $s_{i,j} \theta_1$ are in the same quadrant and the proof is complete.
\end{proof}

\subsection{The algorithm}

We can now describe our ``second stage'' algorithm, which simply runs the Grover schedule $G$ defined as follows.
\begin{definition}
The Grover schedule $G$ comprises the sequence $R = (s_{i,j})_{i = 0 \dots L,\ j = 0 \dots C + 1}$ and $T = (\ceil*{A\log_2(1/(\delta' \theta' \eps' u_i))})_{i = 0 \dots L,\ j = 0 \dots C + 1}$. Here $A$ is a universal constant to be chosen later.
\end{definition}

Note that the $T_{i, j}$ values we use are exactly the number of coin flips used in the second stage of the algorithm in \cite{AR20}. Like in their algorithm, this choice of values allows us to avoid stray $\log(1/\epsilon)$ or $\log\log(1/\epsilon)$ factors in the overall query complexity.

\begin{proposition}                                        \label{prop:out-query}
    Performing the Grover schedule $G$ takes at most $O(\log(1/\delta')/(\theta'\eps'))$ queries.
\end{proposition}
\begin{proof}
    Using \Cref{fact:schedule-calcs}, the query complexity of performing $G$ is ${\sum_{i = 0}^{L}\log(1/(\delta'\theta'\eps' \cdot 1.01^i)) 1.01^i}$ (up to constant factors). This is ${\log(1/\delta)\sum_{i = 0}^{L} 1.01^i + \sum_{i = 0}^{L} \log(1/(\theta'\eps' \cdot 1.01^i)) 1.01^i}$. Noting that ${1.01^L = O(1/(\theta'\eps'))}$, the first term is clearly ${O(\log(1/\delta)/(\theta'\eps'))}$ and the second term, up to a constant factors, is ${\sum_{i = 0}^{L} (L - i) 1.01^i = O(1.01^L) = O(1/(\theta'\eps'))}$. Therefore, the overall query complexity is ${O(\log(1/\delta)/(\theta'\eps'))}$ as desired.
\end{proof}
\begin{regularRemark}
%\begin{remark}
 The above calculation mirrors the  one done for stage~$2$ of the adaptive algorithm in \cite{AR20}; this is expected because both algorithms use the same number of ``coin flips'' per coin ($T'_{i, j}$ values), as mentioned above.
%\end{remark}
\end{regularRemark}

It now remains for us to show how to approximate $\theta^*$ (with high probability) using the data collected from this Grover schedule.

\subsection{Completing the algorithm}

We first prove a lemma showing that we can distinguish between any pair of angles (that are not already sufficiently close to each other) by using the ideas developed in \Cref{lem:split} and \Cref{lem:disting}.

\begin{lemma}                                     \label{lem:disting}
    There is an $O(1)$-time classical deterministic algorithm that, given
     \begin{itemize}
        \item $\theta_0 \apx{1.1} \theta'$, $\theta_1 \apx{1.1} \theta'$, such that $\theta_0 \notapx{1+\eps/6.1} \theta_1$
        \item the data collected by the Grover schedule $G$,
     \end{itemize}
    outputs either ``reject~$\theta_0$'' or ``reject~$\theta_1$''.
    Except with failure probability at most $c \theta' \delta' \eps' / |\theta_1 - \theta_0|$ (where $c > 0$ is a constant to be chosen later), the following is true:
    \[
        \text{For $b = 0, 1$, if  $\theta^* \apx{1 + .001\eps} \theta_b$, then the algorithm does not output ``reject~$\theta_b$''.} 
    \]
\end{lemma}
\begin{proof}
    The algorithm computes the $(i, j)$ pair promised by \Cref{lem:refine}, such that $s_{i,j} \theta_0$ and $s_{i,j}\theta_1$ are in the same quadrant and such that $|s_{i,j} \theta_1 - s_{i,j} \theta_0| \apx{1.5} \frac{\pi}{8}$. Letting $q_b = \sin^2(s_{i,j} \theta_b)$ for $b = 0, 1$, it follows from the assumptions in the preceding sentence that $|q_0 - q_1| \geq .04$. The algorithm may now select a threshold $q' \in [.01, .99]$ such that (without loss of generality) $q_0  \leq q' - .01$ and $q_1 \geq q' + .01$.

    The algorithm will use just the coin flips from the $s_{i,j}$ part of the schedule; these coin flips have bias $p(s_{i,j}) = \sin^2(\llfloor s_{i,j} \rrfloor \theta^*)$.
    More precisely, the algorithm will output ``reject $\theta_0$'' if $\wh{\bp}_{i,j} > q'$ and ``reject $\theta_1$'' if $\wh{\bp}_{i,j} \leq q'$.
    We need to show that if $\theta^* \apx{1 + .001 \eps} \theta_0$ then the algorithm outputs ``reject $\theta_0$'' with probability at most $c \theta' \delta' \eps' / |\theta_1 - \theta_0|$.
    (The case when $\theta^* \apx{1 + .001 \eps} \theta_1$ is analogous.)

    Now if $\theta^* \apx{1 + .001 \eps} \theta_0$, then $\llfloor s_{i,j} \rrfloor \theta^* \apx{1 + .001 \eps} \llfloor s_{i,j} \rrfloor  \theta_0$.
    It follows that
    \[
        \Bigl|\llfloor s_{i,j} \rrfloor \theta^* - \llfloor s_{i,j} \rrfloor  \theta_0\Bigr| \leq .001 \eps \cdot \llfloor s_{i,j} \rrfloor  \theta_0 \leq .001 s_{i,j} \cdot \eps  \theta_0.
    \]
    But we know that
    \[
        \tfrac{\pi}8 \apx{1.5} |s_{i,j} \theta_1 - s_{i,j} \theta_0|  = s_{i,j} |\theta_1 - \theta_0| \geq s_{i,j} \cdot (\eps/6.1) \theta_0,
    \]
    the last inequality because $\theta_0 \notapx{1+\eps/6.1} \theta_1$.
    Combining the above two deductions yields
    \[
        \Bigl|\llfloor s_{i,j} \rrfloor \theta^* - \llfloor s_{i,j} \rrfloor  \theta_0\Bigr| \leq .001 \cdot 1.5 \cdot 6.1 \cdot \tfrac{\pi}{8} \leq .004.
    \]
    Moreover, $\llfloor s_{i,j} \rrfloor  \theta_0$ and $s_{i,j} \theta_0$ differ by at most $2\theta_0 \leq .0002 < .001$.
    Thus we finally conclude
    \[
        \Bigl|\llfloor s_{i,j} \rrfloor \theta^* - s_{i,j} \theta_0\Bigr| \leq .005 \quad \implies \quad \Bigl|p(s_{i,j}) - q_0\Bigr| \leq .005 \quad \implies \quad p(s_{i,j}) < q' - .005
    \]
    (the first implication using that $\sin^2$ is $1$-Lipschitz).
    Then, using a Chernoff bound, we have that $\wh{\bp}_{i,j} > q'$ with probability at most $c\delta'\theta'\eps'u_i$, where $c$ is a constant that depends on $A$ (as $A$ increases, $c$ decreases). From \Cref{lem:split}, we know that $u_i \apx{1.01} \pi/(8|\theta_1 - \theta_0|)$. Then, assuming the constant~$A$ is chosen large enough as a function of~$c$, we indeed have that $\wh{\bp}_{i,j} > q'$ with probability at most $c \theta' \delta' \eps' / |\theta_1 - \theta_0|$.
\end{proof}

\subsubsection{Description of the post-processing algorithm}

We have developed the necessary tools to describe and justify the classical post-processing algorithm.

Fix values $\theta_i = \theta'(1+.001\eps')^i$ for $-V \leq i \leq V$, where $V = O(1/\eps')$ is a minimal integer such that $(1 + .001\eps')^{V} \geq 1.11$. We will refer to the $\theta_i$'s as ``nodes''. We may also assume that the number of nodes is a power of $2$ for convenience (while there are $2|V| + 1$ nodes, we can always pad these actual nodes with some dummy nodes to reach the nearest power of~$2$).

The main idea is to repeatedly use \Cref{lem:disting} to run a ``tournament'' amongst all nodes. The tournament is structured as a series of ``rounds''. In a given round, suppose we start off with $n$ nodes. Sort the nodes in order of the angles they correspond to. Now, pair up node $1$ with node $n/2 + 1$, node $2$ with node $n/2 + 2$, and so on, until node $n/2$ is paired up with node $n$. For each pair of nodes, use \Cref{lem:disting} to choose a winner to go to the next round. Note that it is possible that two nodes that are matched in the tournament do not satisfy the pre-condition of \Cref{lem:disting} because they are within a factor of $(1 + \eps/6.1)$ of each other --- we call these ``void'' match-ups. We will call the part of the tournament we have described so far the first phase --- when we see a void match-up, we stop this first phase and enter the second phase. (Note that if we never see any void match-ups and there is only one node left, we do not need the second phase and can directly output the remaining node as our estimate.)

When the first phase ends, take all remaining nodes and enter the second phase. In this phase, match up every pair of remaining nodes, and for every pair that does not form a void match-up, eliminate one of the nodes using \Cref{lem:disting}. At the end of this, output any one of the remaining un-eliminated nodes (if there are none, then the program can output an arbitrary node - this is a failure condition and we will show that, with high probability, such failure conditions will not happen).

In our algorithm, we have arranged for $\theta_{-V} \leq \theta^* \leq \theta_V$, and hence there exists a node $\theta_{i^*}$ such that $\theta_{i^*} \apx{1+.001\eps} \theta^*$. We will proceed to bound the over failure probability of the algorithm by the probability that this node loses any match-up it is a part of.

\begin{lemma}
\label{lem:NotLosingIsSufficient}
If $\theta_{i^*}$ never loses any match-up it is a part of, then the tournament outputs an estimate $\theta_{\mathrm{est}}$ such that $\theta_{\mathrm{est}} \apx{1 + \eps/6} \theta^*$.
\end{lemma}

\begin{proof}
Suppose that $\theta_{i^*}$ never loses any match-up it is a part of. If the tournament ends in the first phase itself, then the algorithm will output $\theta_{i^*}$, which is correct. If the tournament ends in the second phase, then the only possible other nodes that we could output are the ones that $\theta_{i^*}$ did not play, which can be at most a factor of $1 + \eps/6.1$ away from it. Therefore, these nodes are also at most a factor of $(1 + 0.001\eps)(1 + \eps/6.1) \leq 1 + \eps/6$ away from $\theta^*$.
\end{proof}

\begin{lemma}
\label{lem:boundOnFirstPhase}
The probability that $\theta_{i^*}$ loses any match-up in the first phase of the tournament can be upper-bounded by $c\delta'$ times a constant factor (where $c$ is the constant from \Cref{lem:disting}).
\end{lemma}

\begin{proof}
    Consider an arbitrary round in the first phase of the tournament, where we have $2^j$ nodes remaining for some $j$. By how we chose the match-ups, we know that every two angles that are matched up in that round must be at least a factor of $(1 + 0.001\eps')^{2^{i - 1}}$ apart. This implies that their absolute difference is at least $(\theta'/1.1) \cdot ((1 + 0.001\eps')^{2^{i - 1}} - 1)$. Then, by \Cref{lem:disting}, the failure probability of any match-up in that round is at most  $1.1c\delta'\eps'/((1 + 0.001\eps')^{2^{i - 1}} - 1)$.
    
    Suppose the tournament begins with $n$ nodes. We can use the union bound to upper-bound the probability that $\theta_{i^*}$ loses in any of the at most $\log_2 n$ rounds by
    \[\sum_{j = 1}^{\log_2 n} \frac{1.1c\delta'\eps'}{(1 + 0.001\eps')^{2^{i - 1}} - 1} \leq  \sum_{j = 1}^{\log_2 n} \frac{1.1c\delta'\eps'}{0.001\eps2^{i - 1}} \leq 2200c\delta' \]
\end{proof}

\begin{lemma}
\label{lem:boundOnSecondPhase}
The probability that $\theta_{i^*}$ loses any match-up in the second phase of the tournament can be upper-bounded by $c\delta'$ times a constant factor (where $c$ is the constant from \Cref{lem:disting}).
\end{lemma}

\begin{proof}
    If the algorithm enters the second phase, this means that there is at least one void match-up. Let the number of nodes at this point be $n$. Then all $n/2$ nodes between the pair of nodes involved in the void match-up must be within a factor of $(1 + \eps'/6.1)$ of the first node in the void pair. Therefore, there are at most $2\log_{1 + 0.001\eps'}(1 + \eps'/6.1) \leq 12200$ nodes in total.
    
    Every node that $\theta_{i^*}$ plays in a match-up is at least a factor of $1 + \eps'/6$ away from it, which means that the absolute difference in value between the nodes is at least $(\theta'/1.1) \cdot (\eps'/6)$. By \Cref{lem:disting}, this implies that the failure probability is at most $6.6c\delta'$. Since there are at most $12200$ such match-ups, the overall failure probability is at most $81000c\delta'$ by the union bound.
\end{proof}

We can now prove \Cref{lem:estimateToEps}, which we restate below.

\begin{L9}                                     
    Given the parameters $\theta’$, $\eps’$, and $\delta’$, there is a nonadaptive algorithm that performs $O(\log(1/\delta')/(\theta'\eps'))$ queries, and outputs a result $\theta_{\mathrm{est}}$ with the following guarantee: if $\theta’ \apx{1.11} \theta^*$, then $\theta_{\mathrm{est}} \apx{1 + \eps'/6} \theta^*$ except with probability at most $\delta/2$.
\end{L9}

\begin{proof}
    The algorithm achieving this lemma involves running the Grover schedule $G$ and then post-processing using the tournament algorithm which we described. By \Cref{lem:NotLosingIsSufficient}, the failure probability of the algorithm is bounded by the probability that $\theta_{i^*}$ loses. This is bounded by the probability that it loses in the first phase or the second phase; by \Cref{lem:boundOnFirstPhase} and \Cref{lem:boundOnSecondPhase}, this is at most a constant times $c\delta'$. By choosing $c$ to be sufficiently small, we can successfully make this at most $\delta / 2$.
\end{proof}

\newpage

\bibliography{nonadaptive-grover4}

\end{document}